\documentclass[12pt]{amsart}
\usepackage{amsmath,amsfonts,amssymb} \usepackage{color}
\usepackage[all,cmtip]{xy}
\usepackage{graphicx}

\newcommand{\corr}[1]{\langle {#1} \rangle}

\newcommand{\bg}{{\bf g}}  
\newcommand{\bt}{{\bf t}}

  \newcommand{\cD}{\mathcal{D}} \newcommand{\cE}{\mathcal{E}}

 \newcommand{\bZ}{\mathbb{Z}}
 
 \newcommand{\pd}{\partial}
\newcommand{\Mbar}{\overline{\mathcal M}}

\newcommand{\tz}{\tilde{z}}

\DeclareMathOperator{\Res}{Res} 

\newcommand{\be}{\begin{equation}}
\newcommand{\ee}{\end{equation}}
\newcommand{\bea}{\begin{eqnarray}}
\newcommand{\eea}{\end{eqnarray}}
\newcommand{\ben}{\begin{eqnarray*}}
\newcommand{\een}{\end{eqnarray*}}

\newcommand{\half}{\frac{1}{2}}

\newtheorem{cor}{Corollary}[section]

 \newtheorem{prop}[cor]{Proposition}
 \newtheorem{thm}[cor]{Theorem}
\theoremstyle{remark}

\definecolor{A}{rgb}{.75,1,.75}

\definecolor{yellow}{rgb}{1,1,0}
\definecolor{orange}{rgb}{1,.7,0}
\definecolor{red}{rgb}{1,0,0}
\definecolor{white}{rgb}{1,1,1}

 \makeindex
\begin{document}
\title
{Grothendieck's Dessins d'Enfants in a Web of Dualities. II.}

\author{Jian Zhou}
\address{Department of Mathematical Sciences\\Tsinghua University\\Beijing, 100084, China}
\email{jianzhou@mail.tsinghua.edu.cn}

\begin{abstract}
We show that the spectral curve for Eynard-Orantin topolgical recursions satisfied by
counting Grothendieck's dessins d'enfants are related to Narayana numbers. 
This suggests a connection of dessins to combinatorics of Coxeter groups, 
noncrossing partitions, free probability theory,
and cluster algebras.
\end{abstract}

\maketitle

\section{Introduction}

In the first part \cite{Zhou-Des} of this series of papers,
we have proposed to study dualities of different topological quantum field theories
using   symmetries of the KP hierarchy.
Such theories involve infinitely many formal variables,
so it is natural to study them altogether from a point of view of infinite-dimensional geometry.
One of the most developed approaches to infinite-dimensional geometry is the theory
of integrable hierarchies.
The Kyoto school's approach to integrable hierarchies focuses on their symmetries,
described by the actions infinite-dimensional Lie groups or Lie algebras.
In this approach the KP hierarchy is universal in the sense that
its symmetries are described by $\widehat{GL}(\infty)$,
into which the symmetry groups for other integrable hierarchies
can be embedded.

From this point of view,
different theories have equal status.
Once the partition function of a particular theory is identified
as a tau-function of the KP hierarchy,
it can be transformed into the partition function of any other theory,
which is also a tau-function of the KP hierarchy.
This is very crucial to our understanding of dualities
of different theories,
because the constructions of two theories may be totally unrelated,
but they should share exactly the same properties if they are dual to each other
in our sense.

However, the democracy among different theories does not mean
that some specific theories cannot play more important roles than the others.
The democracy in this setting only means that before some further studies,
each theory is just a candidate for some selection process
after which one may or may not determine some specific theories that
will play a unifying and leading role in the future studies.
One of the goals of this series of papers is to show that
counting of Grothendieck's dessins of d'enfants
plays such a role.
In Part I of this series,
we have focused on some enumeration problems.
We discover that
various enumeration problems
are either special cases of the enumeration of dessins,
or are directly related to it.
We have seen that their generating series
are tau-functions of the KP hierarchy.
These tau-functions have very special properties.
For example,
they satisfy suitable versions of Virasoro constraints,
and  admit some cut-and-join representations.
It is often possible to find explicit formulas for the bases
of the elements  in Sato's Grassmannian corresponding to these tau-functions,
sometimes with the help of representation theory.
With these bases available
one can find the Kac-Schwarz operators.
Furthermore,
one can use the first basis vectors to write down the corresponding quantum spectral curves.
To make contact with mirror symmetry,
we have reformulated the differential equations involved in these examples
of quantum spectral curves as differential equations
of hypergeometric type.

In this Part II,
we will still focus on the examples considered in Part I,
but we will examine them  further from the point of view of emergent geometry of KP hierarchy
developed by the author in earlier works.

Emergence is an notion developed in statistical physics.
Roughly it means the appearance of some structures
when one deals with a system of very large degrees of freedom.
We have borrowed this notion to describe some results in the study of Gromov-Witten type theories.
Such a theory has a finite-dimensional phase space
consisting of the primary operators.
This space is called the small phase space.
Intrinsic to a Gromov-Witten type theory is
its coupling to 2D topological  gravity.
As a result,
each primary operator has an infinite sequence of associated operators
called its gravitational descendants.
The big phase space,
i.e. the space  that includes all the primaries and their descendants,
is then infinite-dimensional.
We refer to the geometric structures that can be more easily seen from the big phase space
than from the small phase space as the emergent geometry.
We emphasize that  such structures also exists over the small phase space,
but it is more natural to think of them as the restrictions of some structures over
the big phase space.

The emergent approach to the study of a Gromov-Witten type theory is
then to look for some structures over the big phase space
and consider their restrictions to the small phase space to
get back to a finite-dimensional setting.
This is in contrast to the reduction/reconstruction approach
which starts with some structures on the small phase space
and tries to reconstruct the whole theory based on some axioms
that characterize such theories.

In earlier papers,
we have shown that in some examples of Gromov-Witten type theories,
the following structures emerge:
a plane algebraic curve called the spectral curve,
a family of deformations of this curve defined over the big phase space called its special
deformation,
and a conformal field theory defined over the spectral curve.
More recently,
it has been realized that Eynard-Orantin topological recursions
can also be understood as emergent geometry.
In summary,
emergent geometry is a synthesis of ideas from mirror symmetry, Witten Conjecture/Kontsevich Theorem,
and Eynard-Orantin topological recursions.

It is possible to study the emergent geometry of any Gromov-Witten type theory whose partition
function is a tau-function of the KP hierarchy \cite{Zhou-KP}.
It is particularly simple when the partition function also satisfies the Virasoro constraints.
The original motivation was to understand the mirror symmetry of a point,
more precisely,
why the $n$-point functions associated with the Witten-Kontsevich tau-function \cite{Wit, Kon}
satisfied the Eynard-Orantin topological recursion
on the Airy curve $y^2 = x$,
a result proved many times by different approach in e.g. \cite{EO, BCSW, Zhou-LMP}.
The proof by the author in \cite{Zhou-LMP} indicates that the Eynard-Orantin topological recursion
in this case is equivalent to the Virasoro constraints derived by Dijkgraaf-Verlinde-Verlinde \cite{DVV},
and spectral curve and the Bergman  kernel are determined by making sense of
genus zero one-point function and the genus zero two-point function respectively.
The emergent geometry in this sense of 
the modified partition functions of Hermitian one-matrix models with even coupling constants
has been studied in \cite{Zhou-Even}.
The same has been done for partition functions of Hermitian one-matrix models
in \cite{Zhou-Matrix}.
More recently,
in Part I of this series \cite{Zhou-Des} we have discovered that the enumeration of Grothendieck's dessins d'enfants
lies in the center of a web of dualities of various enumerative problems,
including the above three cases.
One of the goals of this paper is to present the emergent geometry
of Grothendick's dessins d'enfants
in the same fashion as in the above three cases.

It may seem redundant to consider the Eynard-Orantin topological recursions in these cases
because it is concerned with the computations of genus $g$, $n$-point functions recursively,
but on the other hand we have already obtained a formula for all
$n$-point functions, summed over all genera,
in \cite{Zhou-Des} based on the general results for tau-functions of the KP hierarchy developed in \cite{Zhou-KP}.
Our reason is that these different  approaches yield different  formulas that encode different information.
The formulas in Part I  encode information related to problems in representation theory,
while the formulas in this Part encode the information related to  combinatorial problems.
Our second goal of this paper is to present some unexpected connection between counting dessins
and Narayana numbers.
More precisely,
the genus zero one-point function of Grothendieck's dessins
is the generating function of the Narayana numbers.
Because the latter is related to cluster algebras, root systems,
Coexter groups, associahedra, etc, see e.g. \cite{Fom-Rea},
this suggests a connections from dessins to these mathematical objects.
In particular,
because Grothendieck's dessins correspond to Type A root systems,
our results suggest the possibility of dessin theory of Type BCD.
Note in our previous works on emergent geometry,
we encounter only Catalan numbers.
Because Narayana numbers refine Catalan numbers,
we get another evidence that counting dessins is more fundamental than counting fat graphs
(aka clean dessins).

Derivation of topological recursion from Virasoro constraints for counting Grothendieck' dessins
has already been given by Kazarian and Zograf \cite{Kaz-Zog}.
But their version is not exactly the version of Eynard and Orantin \cite{EO}.
They have remarked that one can derive the latter version from their version. 
We will present a direct derivation of the Eynard-Orantin version
in the same fashions as in \cite{Zhou-Even, Zhou-Matrix}.
Our motivation is, 
first of all,
to interpret the Eynard-Orantin recursions in these cases 
as emergent geometry. 
Secondly,
as a consequence,
we establish a connection between counting dessins and 
intersection theory on $\Mbar_{g,n}$.

We arrange the rest of the paper in the following way.
In Section  \ref{sec:Catalan} we recall some earlier results on emergent geometry of some enumeration 
problems and the appearance of Catalan numbers in their genus zero one-point functions.
We show how to compute the $n$-point functions in arbitrary genera for dessin tau-function
using the dessin Virasoro constraints in  Section 3.
In particular,
the genus zero one-point function in this case is shown to be the generating series of
the Narayana numbers.
The Eynard-Orantin topological recursions are shown to be equivalent to 
the dessin Virasoro constraints in Section \ref{sec:Emergent}.
In Section \ref{sec:Intersection}
wee present some computations necessary for relating dessins to intersection
numbers numbers using Eynard's results \cite{Eynard-Int2}.
In he final Section \ref{sec:Conclusion}
we present some concluding remarks.

\section{Earlier Results on Emergent Geometry and Catalan Numbers}

\label{sec:Catalan}

In this Section we recall some results on emergent geometry in the following three cases:
Witten-Kontsevich tau-function,
partition functions of Hermitian one-matrix models,
and modified partition functions of Hermitian one-matrix models with even coupling constants.
We will focus on the genus zero one-point functions
and two-point functions.
We note that the genus zero
one-point functions in these three cases
are all related to the Catalan numbers.

\subsection{Emergent geometry of Witten-Kontsevich tau-function}

The first example of emergent geometry of Gromov-Witten type theories
appeared in \cite{Zhou-Mirror},
based on \cite{ADKMV}.
In genus zero this means the construction of the special deformation of the Airy curve
defined using the genus zero one-point function associated to the Witten-Kontsevich tau-function \cite{Zhou-Mirror}.
In higher genera,
this means an equivalent reformulation of the Virasoro constraints
in terms of quantization of a field on the Airy curve \cite{Zhou-Mirror},
or Eynard-Orantin recursions on the Airy curve \cite{Zhou-LMP}.

The relevant correlators are
\be
\corr{\tau_{a_1} \cdots \tau_{a_n}}_g = \int_{\Mbar_{g,n}} \psi_1^{a_1} \cdots \psi_n^{a_n}.
\ee
The genus $g$ free energy is defined by
\be
F_g = \sum_n \frac{t_{a_1} \cdots \tau_{a_n}}{n!} \corr{\tau_{a_1} \cdots \tau_{a_n}}_g,
\ee
where the summation is taken over $n$ such that $2g-2+n > 0$.
The Virasoro constraints in terms of the free energy are given by:
\bea
&& \frac{\pd F^{WK}}{\pd  g_0}
= \sum_{k=1}^\infty (2k+1) g_k\frac{\pd F^{WK}}{\pd g_{k-1}} + \frac{g_0^2}{2\lambda^2},  \label{eqn:Virasoro-1} \\
&& \frac{\pd F^{WK}}{\pd g_1}
= \sum_{k=0}^\infty (2k+1) g_k\frac{\pd F^{WK}}{\pd g_{k}} + \frac{1}{8},  \label{eqn:Virasoro0}\\
&& \frac{\pd F^{WK}}{\pd u_n}
= \sum_{k=0}^\infty (2k+1) g_k \frac{\pd F^{WK}}{\pd g_{k+n-1}} \label{eqn:VirasoroN} \\
&& \;\;\;\;\;\; + \frac{\lambda^2}{2} \sum_{k=0}^{n-2} \biggl(\frac{\pd^2 F^{WK}}{\pd g_k\pd g_{n-k-2}}
+ \frac{\pd F^{WK}}{\pd g_k} \frac{\pd F^{WK}}{\pd g_{n-k-2}} \biggr), \;\;\; n \geq 2. \nonumber
\eea
Usually $\{t_n\}_{n\geq 0}$ are used to denote the coupling constants.
Here we have made the following change of coordinates:
\be
t_k = (2k+1)!! g_k.
\ee
Correspondingly,
the operator $\tau_n$ is changed to $\sigma_n = (2n+1)!!\tau_n$.

Because $\Mbar_{0,1}$ and $\Mbar_{0,2}$ are not defined,
the genus zero one-point function and  two-point function  need special care.
One approach,
taken in \cite{Zhou-LMP} and elaborated further in \cite{Zhou-KP},
is to introduce ghost operators $\tau_{-n}$ and ghost variables $t_{-n}$.
Another approach,
taken in \cite{Zhou-Mirror},
is to consider the theory in a different background,
i.e., with nonzero $g_0$.
To compute the genus zero one-point function in this background,
one then needs to compute:
\ben
\corr{\sigma_n e^{g_0 \sigma_0}}_0
& = & \frac{g_0^{n+2}}{(n+2)!}\corr{\sigma_0^{n+2}\sigma_n}_0
=   \frac{(2n+1)!! }{(n+2)!} g_0^{n+2} \\
& = & \frac{1}{n+2} \binom{2n+2}{n+1} \cdot \frac{g_0^{n+2}}{2^{n+1}},
\een
where
\be
C_n = \frac{1}{n+1} \binom{2n}{n}
\ee
are the {\em Catalan numbers}.
We note the above result can be reformulated as follows:
\be \label{eqn:Der-F0}
z (1- \frac{2g_0}{z^2})^{1/2} = z - \frac{g_0}{z}
- \sum_{n =0}^\infty \frac{\pd F_0}{\pd g_n}(g_0, 0, \dots) \cdot z^{-(2n+3)}.
\ee
So we define the genus zero one-point function
associated to Witten-Kontsevich tau-function by:
\be
G^{WK}_{0,1}(z; g_0) = \sum_{n =0}^\infty \frac{\pd F_0}{\pd g_n}(g_0, 0, \dots) \cdot z^{-(2n+3)}
= z - \frac{g_0}{z} - z (1- \frac{2g_0}{z^2})^{1/2}.
\ee
In \cite{Zhou-Mirror},
motivated by \eqref{eqn:Der-F0},
we consider the following series:
\be \label{eqn:X-in-F}
x = z - \sum_{n \geq 0} (2n+1) g_n z^{2n-1} - \sum_{n \geq 0} \frac{\pd F_0}{\pd g_n}(\bg) \cdot z^{-2n-3}.
\ee
By Virasoro constraints one has
\be \label{eqn:x^2}
\begin{split}
x^2 = &
2 y \biggl(1- \sum_{n \geq 1} (2 n+1) g_n (2y)^{n-1}\biggr)^2 \\
- & 2 g_0 \biggl(1-  \sum_{n \geq 1} (2n+1) g_n (2y)^{n-1} \biggr) \\
+ & 2 \sum_{n \geq 0} \sum_{k \geq n+2} (2k+1) g_k \cdot \frac{\pd F_0}{\pd g_n} \cdot (2y)^{k-n-2} ,
\end{split}
\ee
where $y =\half z^2$.
In fact,
the following equation is equivalent to Virasoro constraints for $F_0^{WK}$:
\be
(x^2)_- = 0.
\ee
Here for a formal series $\sum_{n \in \bZ} a_n z^n$,
\be
(\sum_{n \in \bZ} a_n z^n)_+ = \sum_{n \geq 0} a_n z^n, \;\;\;\;
(\sum_{n \in \bZ} a_n z^n)_- = \sum_{n < 0} a_n z^n.
\ee

The  genus zero two-point correlators are given by:
\ben
\corr{\sigma_k\sigma_le^{g_0\sigma_0}}_0
& = & \frac{g_0^{k+l+1}}{(k+l+1)!}
\corr{\sigma_k\sigma_l\sigma_0^{k+l+1}}_0
= \frac{(2k+1)!!(2l+1)!!}{k!l!(k+l+1)}g_0^{k+l+1}.
\een
The genus zero two-point function in this case is defined by:
\ben
G_{0,2}(z_1,z_2; g_0)
& = & \sum_{n_1,n_2 =0}^\infty
\frac{\pd^2 F_0}{\pd g_{n_1}\pd g_{n_2}}(g_0, 0, \dots)
\cdot z_1^{-(2n_1+3)} z_2^{-(2n_2+3)} \\
& = & \sum_{k,l =0}^\infty
\frac{(2k+1)!!(2l+1)!!}{k!l!(k+l+1)}g_0^{k+l+1}
\cdot z_1^{-(2k+3)} z_2^{-(2l+3)} .
\een
To take the summation over $k$ and $l$,
we first take derivative in $g_0$ to get:
\ben
\frac{\pd}{\pd g_0} G_{0,2}(z_1,z_2; g_0)
& = & \sum_{k,l =0}^\infty
\frac{(2k+1)!!(2l+1)!!}{k!l!}g_0^{k+l}
\cdot z_1^{-(2k+3)} z_2^{-(2l+3)} \\
& = & \sum_{k=0}^\infty
\frac{(2k+1)!!}{k!}g_0^{k} z_1^{-(2k+3)}   \cdot
\sum_{l =0}^\infty
\frac{(2l+1)!!}{l!}g_0^{l}z_2^{-(2l+3)} \\
& = & \frac{1}{(z_1^2-2g_0)^{3/2}(z_2^2-2g_0)^{3/2}},
\een
and so after integration we get:
\be
\begin{split}
G_{0,2}(z_1,z_2; g_0)
& = \frac{z_1^2+z_2^2-4g_0}{(z_1^2-z_2^2)^2(z_1^2-2g_0)^{1/2}(z_2^2-2g_0)^{1/2}} \\
& - \frac{z_1^2+z_2^2}{(z_1^2-z_2^2)^2z_1z_2}.
\end{split}
\ee

\subsection{Emergent geometry of  Hermitian one-matrix models}

We refer to \cite{Zhou-HMM}, \cite{Zhou-Fat-Thin1}, \cite{Zhou-Fat-Thin2} for notations and proofs of the results in this Subsection.
Roughly speaking,
with the introduction of the t' Hooft coupling constant $t= Ng_s$,
one gets the free energy functions $F_g(\bt)$ from the fat genus expansions
of Hermitian one-matrix models.
The genus zero fat Virasoro constraints can be reformulated as follows:
The series
\ben
&& y = \frac{1}{2} \sum_{n=0}^\infty
\frac{t_n -\delta_{n,1}}{n!} x^n
+ \frac{t}{x} +\sum_{n = 1}^\infty
\frac{n!}{x^{n+1}} \frac{\pd F_{0}(\bt)}{\pd t_{n-1}}
\een
is a solution to the equation
\be
(y^2)_- = 0.
\ee
When all $t_n$ are taken to be $0$,
\ben
y & = & - \frac{x}{2} + \frac{t}{x} + \frac{t^2}{x^3}+\frac{2t^3}{x^5}+\frac{5t^4}{x^7}+\frac{14t^5}{x^9}
+ \cdots \\
& = & - \frac{x}{2} + \sum_{n=1}^\infty \frac{1}{n} \binom{2n-2}{n-1} \frac{t^n}{x^{2n-1}},
\een
so we again encounter Catalan numbers. Note
\be \label{eqn:Fat-Spec1}
y = \frac{-\sqrt{x^2-4t}}{2},
\ee
and so the spectral curve in terms of the field $y=y(z)$ is given by the algebraic curve:
\be \label{eqn:Fat-Spec}
x^2 - 4y^2 = 4t.
\ee

Define the genus zero one-point function in this case by:
\be
G_{0,1}(x) = \frac{t}{x} + \sum_{m\geq 1}   x^{-m-1} \cdot m\frac{\pd F_g}{\pd g_{m}}\biggl|_{g_i = 0},
\ee
then the above result shows that
\be
G_{0,1}(x) = \frac{1}{2} (x- \sqrt{x^2-4t}).
\ee
In terms of the correlators,
\be
\corr{p_m}_0^c(t)
= \begin{cases}
\frac{1}{n+1}\binom{2n}{n} t^{n+1} & \text{if $m =2n$}, \\
0 & \text{if $m = 2n+1$}.
\end{cases}
\ee
The genus zero two-point function is computed in \cite{Zhou-Matrix} to be:
\be
 G_{0,2}(x_1,x_2)
= \frac{x_1x_2-4t}{2(x_1-x_2)^2\sqrt{(x_1^2-4t)(x_2^2-4t)}}
- \frac{1}{2(x_1-x_2)^2}.
\ee
It has the following expansion:
\be \label{eqn:G02-2}
\begin{split}
 G_{0,2}(x_1,x_2)
= &  \sum_{m,n=0}^\infty  \frac{t^{m+n+1}}{m+n+1} \frac{(2m+1)!}{(m!)^2}
\cdot   \frac{(2n+1)!}{(n!)^2} \frac{1}{x_1^{2m+2}x_2^{2n+2}} \\
+ & \sum_{m,n=0}^\infty  \frac{4 t^{m+n+2}}{m+n+2} \frac{(2m+1)!}{(m!)^2}
\cdot   \frac{(2n+1)!}{(n!)^2} \frac{1}{x_1^{2m+3}x_2^{2n+3}}.
\end{split}
\ee
In terms of correlators we have
\be
\corr{p_{2m+1}p_{2n+1} }_0^c(t) =  \frac{(2m+1)!}{(m!)^2} \cdot \frac{(2n+1)!}{(n!)^2} \cdot \frac{t^{m+n+1}}{m+n+1},
\ee
\be
\corr{p_{2m+2}p_{2n+2} }_0^c(t) =  \frac{(2m+1)!}{(m!)^2} \cdot \frac{(2n+1)!}{(n!)^2} \cdot \frac{4t^{m+n+2}}{m+n+2},
\ee
and other genus zero two-point correlators vanish.

\subsection{Emergent geometry of  modified partition function}

Now we recall the case of modified partition function of Hermitian one-matrix model with
even couplings.
For details, see \cite{Zhou-Even}.
The Virasoro constraints in genus zero are
the following sequence of differential equations:
\bea
&& \half \frac{\pd F_0}{\pd s_2} =  \sum_{k\geq 1}
k s_{2k} \frac{\pd F_0}{\pd s_{2k}} + \frac{t^2}{4},  \label{eqn:L0-0} \\
&& \half \frac{\pd F_0}{\pd s_{2n+2}} =
\sum^{n-1}_{k=1} \frac{\pd F_0}{\pd s_{2k}}
\frac{\pd F_0}{\pd s_{2n-2k}}
+ t \frac{\pd F_0}{\pd s_{2n}} +  \sum_{k\geq 1}
ks_{2k} \frac{\pd F_0}{\pd s_{2k+2n}}, \quad n \geq 1.
\eea
In this case,
we consider:
\be \label{eqn:Special}
y:=\half \sum_{k=1}^\infty k(s_{2k}-\half \delta_{k,1}) x^{k-1}
+ \frac{t}{2x}  + \sum_{k=1}^\infty \frac{1}{x^{k+1}}
\frac{\pd F_0}{\pd s_{2k}}.
\ee
Then  by the Virasoro constraints above:
\be \label{eqn:y2minus}
(y^2)_- = \biggl( \half t(s_{2}-\half)
+ \sum_{l \geq 1}(l+1) s_{2l+2}
\frac{\pd F_0}{\pd s_{2l}} \biggr) x^{-1},
\ee
and so
\ben
y^2 & = & \frac{1}{4}
(\sum_{k=1}^\infty k(s_{2k}-\half \delta_{k,1}) x^{k-1})^2
 \\
& + & \frac{t}{2} \sum_{k=1}^\infty k(s_{2k}-\half \delta_{k,1}) x^{k-2}
+ \sum_{l \geq 1} \sum_{k \geq l+1} k s_{2k}
\frac{\pd F_0}{\pd s_{2l}} x^{k-l-2}.
\een
It follows that when all $s_{2k} = 0$,
we get the spectral curve:
\be \label{eqn:Spectral}
y^2=\frac{1}{16}-\frac{t}{4x}.
\ee
The one-point function in genus zero is:
\be \label{eqn:G01}
G_{0,1}(x)
= \frac{1}{4}
\biggl(1-\frac{2t}{x}
- \sqrt{1-\frac{4t}{x}}\biggr).
\ee
It has the following expansion:
\be
G_{0,1}(x) = \frac{1}{4}
\sum_{n \geq 2} \frac{(2n-3)!!}{n!}
\cdot (2t)^{n}  x^{-n}
= \frac{1}{2} \sum_{n \geq 2} \frac{(2n-2)!}{(n-1)!n!} \frac{t^n}{x^n}.
\ee
The two-point function in genus zero is:
\be \label{eqn:G02}
 G_{0,2}(x_1,x_2)
=  \frac{1-\frac{2t}{x_1}-\frac{2t}{x_2}}
{2(x_1-x_2)^2\sqrt{(1-\frac{4t}{x_1})(1-\frac{4t}{x_2})}}
- \frac{1}{2(x_1-x_2)^2}.
\ee
It has the following expansion:
\be \label{eqn:G02-2}
 G_{0,2}(x_1,x_2)
= 2 \sum_{l=2}^\infty  \frac{t^l}{l }
\sum_{m+n=l-2}  \frac{(2m+1)!}{(m!)^2}
\cdot   \frac{(2n+1)!}{(n!)^2}
\frac{1}{x_1^{m+2}x_2^{n+2}},
\ee
In terms of correlators we have
\be
\corr{p_{2m}p_{2n} }_0^c = \frac{1}{2} \cdot \frac{(2m)!}{(m-1)!m!} \cdot \frac{(2n)!}{(n-1)!n!} \cdot \frac{t^{m+n}}{m+n}.
\ee

\section{$N$-Point Functions of Dessin Tau-Function by Virasoro Constraints}

\label{sec:N-Point}

In this Section we show how to compute  the $n$-point functions associated to the dessin tau-function
 by the dessin Virasoro constraints.

\subsection{Dessin   $n$-point functions}

Recall the genus $g$, $n$-point functions associated to the dessin free energy  $F_{Dessins}$ 
\cite{Zhou-Des}  is defined by:
\be
G^{(n)}(\xi_1, \dots, \xi_n)
= \sum_{m_1, \dots, m_n\geq 1}
\prod_{j=1}^n \frac{m_j}{\xi_j^{-m_j-1}} \frac{\pd}{\pd p_{m_j}}  F_{Dessins} \biggl|_{p_m =0, m \geq 1}
\cdot
\ee
An explicit formula for these functions based on the theory of KP hierarchy
have been proved in \cite{Zhou-Des}.
In particular,
\ben
G^{(1)}(\xi) & = &
\sum_{n \geq 1} \frac{s^nuv}{n} \sum_{i+j=n-1}
 \frac{(-1)^j}{i!j!} \prod_{a=1}^{i} (u+a)(v+a)  \\
 && \cdot \prod_{b=1}^j (u-b)(v-b) \cdot \xi^{-n-1}.
\een
It is actually possible to take the summation $\sum_{i+j=n-1}$
and simplify such a formula.
In this Section we will show how to recursively compute the $n$-point functions
using the dessin Virasoro constraints. 

The  $n$-point correlators are defined by:
\be
\corr{p_{a_1} \cdots p_{a_n}}_g^c
: = a_1\cdots a_n
\frac{\pd^n F_{g,dessin}}{\pd p_{a_1} \cdots \pd p_{a_n}} \biggl|_{p_k =0, k \geq 0}.
\ee
Define the genus $g$, $n$-point function by
\be
G_{g,n}(x_1, \dots, x_n)
= \sum_{a_1, \dots, a_n \geq 0}
\corr{p_{a_1} \cdots p_{a_n}}_g^c
\frac{1}{x_1^{a_1+1} \cdots x_n^{a_n+1}}.
\ee 

\subsection{Virasoro constraints for dessin tau-function}
Recall the Virasoro constraints for dessins are given by \cite{Kaz-Zog}:
For $n \geq 0$,
\be
\begin{split}
L_n  = & -\frac{n + 1}{s} \frac{\pd}{\pd p_{n+1}}
+(u +v)n \frac{\pd}{\pd p_n}
+ \sum_{j=1}^\infty p_j(n + j) \frac{\pd}{\pd p_{n+ j}} \\
+ & \sum_{i + j=n} i j \frac{\pd^2}{\pd p_i\pd p_j}
+\delta_{n,0}uv.
\end{split}
\ee
To keep track of the genus,
introduce the genus tracking parameter.
This corresponds to the following change of variables:
\begin{align}
s\mapsto \lambda s, u \mapsto \lambda^{-1} u, v \mapsto \lambda^{-1} v, p_n \mapsto \lambda^{-1} p_n.
\end{align}
The dessin Virasoro operators then take the following form: 
\be
\begin{split}
L_n  = & -\frac{n + 1}{s} \frac{\pd}{\pd p_{n+1}}
+(u +v)n \frac{\pd}{\pd p_n}
+ \sum_{j=1}^\infty p_j(n + j) \frac{\pd}{\pd p_{n+ j}} \\
+ & \lambda^{2}\sum_{i + j=n} i j \frac{\pd^2}{\pd p_i\pd p_j}
+\delta_{n,0}uv \lambda^{-2}.
\end{split}
\ee

\subsection{Virasoro constraints in terms of correlators}

When expressed in terms of correlators,
the Virasoro constraints can be expressed as follows.
First,
\begin{align}
\corr{p_1}_0 & = suv,
\end{align}
and for $m \geq 0$,
\be \label{eqn:P2m}
\begin{split}
& \frac{m+1}{s} \corr{p_{m+1} \cdot p_{a_1} \cdots p_{a_n}}_{g} \\
= & \sum_{j=1}^n (a_j+m) \cdot \corr{p_{a_1} \cdots
p_{a_j+m} \cdots p_{a_n} }_{g}   \\
+ & (u+v) m\corr{p_{m} \cdot p_{a_1} \cdots p_{a_n}} _{g} \\
+ & \sum_{k=1}^{m-1} k(m-k)\corr{p_{k}p_{m-k} \cdot
p_{a_1} \cdots p_{a_n} }_{g-1}  \\
 + &\sum_{k=1}^{m-1} \sum_{\substack{g_1+g_2=g\\I_1 \coprod I_2 = [n]}}
k\corr{p_{k} \cdot \prod_{i\in I_1} p_{a_i}}_{g_1}
\cdot (m-k) \corr{p_{m-k} \cdot \prod_{i\in I_2} p_{2a_i}}_{g_2},
\end{split}
\ee
where $[n]=\{1, \dots, n\}$.

\subsection{Genus zero one-point function by Virasoro constraints and Narayana polynomials}

From
\ben
\corr{p_1}_0 & = & suv, \\
m\corr{p_{m}}_0^c & = & s \sum_{k=1}^{m-2}
k\corr{p_{k}}_0 \cdot (m-1-k) \corr{p_{m-1-k}}_0^c \\
& + & s(u+v) \cdot (m-1)\corr{p_{m-1}}_{0}, \quad m \geq 2,
\een
we  get:
\ben
G_{0,1}(x) & = &   \frac{suv}{x^2}
+ s\sum_{n \geq 2} \frac{1}{x^{n+1}} \biggl( \sum_{k=1}^{n-2}
k\corr{p_{k}}_0^c \cdot (n-1-k)\corr{p_{n-1-k}}_0 \\
& + & (u+v) \cdot (n-1) \corr{p_{n-1}}_{0} \biggr) \\
& = & \frac{suv}{x^2} + \frac{s(u+v)}{x} G_{0,1}(x)
+ s G_{0,1}(x)^2.
\een
After solving for $G_{0,1}(x)$ we then get:
\be \label{eqn:G01}
\begin{split}
G_{0,1}(x) = & \frac{1}{2s}
\biggl(1-\frac{s(u+v)}{x}
- \sqrt{\biggl(1-\frac{s(u+v)}{x}\biggr)^2
-\frac{4s^2uv}{x^2}}\biggr) \\
= & \frac{1}{2s}
\biggl(1-\frac{s(u+v)}{x}
- \sqrt{1-\frac{2s(u+v)}{x}+\frac{s^2(u-v)^2}{x^2}}\biggr).
\end{split}
\ee
Now recall the following identity:
\be
\begin{split}
& \frac{1}{2} \biggl(1-(u+v)z  - \sqrt{1-2(u+v)z+(u-v)^2z^2}\biggr) \\
= & \sum_{n=1}^\infty \frac{z^{n+1}}{n} \sum_{k=1}^n \binom{n}{k} \binom{n}{k-1} u^{n+1-k} v^k.
\end{split}
\ee
It gives the generating series of the Narayana numbers
\be
N_{n,k} = \frac{1}{k} \binom{n-1}{k-1}\binom{n}{k-1}
= \frac{1}{n} \binom{n}{k}\binom{n}{k-1}.
\ee
They are A001263 on \cite{OEIS}.
The polynomials
\be
N_n(q):=\sum_{k=1}^n N_{n,k} q^k
\ee
is called the $n$-th Narayana polynomial.
The following are the first few terms of $G_{0,1}(x)$:
\ben
G_{0,1}(x)
& = & \frac{suv}{x^2} + \frac{s^2uv(u+v)}{x^3}
+ \frac{s^3uv(u^2+3uv+v^2)}{x^4}\\
& + & \frac{s^4 uv(u^3+6u^2v+6uv^2+v^3)}{x^5} \\
& + & \frac{s^5uv(u^4+10u^3v+20u^2v^2+10uv^3+v^4)}{x^6}
+\cdots.
\een
In terms of the correlators,
we have shown that:
\be
\corr{p_{n}}_0 = s^n uv\sum_{k=1}^n \frac{\binom{n}{k}}{n}
\frac{\binom{n}{k-1}}{n} u^{n-k} v^{k-1}
= \frac{1}{n} s^n u^{n+1}N_n(\frac{v}{u}).
\ee

It is well-known that the Narayana numbers provide a version of $q$-analogues of the Catalan numbers.
By taking $q=1$ in the Narayana polynomials,
one gets the Catalan numbers:
\be
\sum_{k=1}^n \frac{1}{n} \binom{n}{k}\binom{n}{k-1}
= \frac{1}{n+1} \binom{2n}{n}.
\ee
See e.g. \cite{Fom-Rea} and \cite{Arm} for references to related mathematical objects.

\subsection{Computation of genus zero two-point function by Virasoro constraints}

By   equation \eqref{eqn:P2m},
\ben
\corr{p_1 \cdot p_{n}  }_{0}  = sn \cdot \corr{ p_{n} }_{0} ,
\een
and so
\ben
&& \sum_{n=1}^\infty n\corr{p_1 \cdot  p_{n} }_{0} \frac{1}{x_2^{n+1}}
=  \sum_{n \geq 1} sn \cdot \corr{ p_{n}   }_{0}   \frac{n}{x_2^{n+1}} \\
& = & - \frac{d}{dx_2} \biggl(sx_2 \sum_{n \geq 1} \corr{ p_{n}}_{0}  \frac{n}{x_2^{n+1}} \biggr) \\
& = & -\frac{d}{dx_2} \biggl(sx_2 \cdot \frac{1}{2s}
\biggl(1-\frac{s(u+v)}{x_2}
- \sqrt{1-\frac{2s(u+v)}{x_2}+\frac{s^2(u-v)^2}{x_2^2}}\biggr) \biggr) \\
& = & \frac{1-\frac{s(u+v)}{x_2}}{2\sqrt{1-\frac{2s(u+v)}{x_2}+\frac{s^2(u-v)^2}{x_2^2}}} -\frac{1}{2}.
\een
Now recall the following combinatorial identity:
\be
\begin{split}
& \frac{1}{2} \frac{1-(u+v)z}{\sqrt{1-2(u+v)z+(u-v)^2z^2}}
- \frac{1}{2} \\
= & \sum_{n=1}^\infty z^{n+1} \sum_{k=1}^n \binom{n}{k} \binom{n}{k-1} u^{n+1-k} v^k \\
= &  \sum_{n=1}^\infty nz^{n+1}u^{n+1}N_n(\frac{v}{u}).
\end{split}
\ee
It gives the generating series of A132812  on \cite{OEIS}.
The first few terms are:
\ben
\sum_{n=1}^\infty n \corr{p_1 \cdot  p_{n} }_{0} \frac{1}{x_2^{n+1}}
&=&\frac{uvs^2}{x_2^2}+\frac{2uv(u+v)s^3}{x_2^3}
+\frac{3 uv(u^2+3uv+v^2)s^4}{x_2^4} \\
&+ & \frac{4uv(u^3+6u^2v+6uv^2+v^3)s^5}{x_2^5}
+\cdots.
\een

Next we consider the following equation:
\be
2\corr{p_2 \cdot p_{n} }_{0}
 = s (n+1) \cdot \corr{ p_{n+1} }_{0}
 + s (u+v) \cdot \corr{p_1 \cdot p_{n}}_{0} .
\ee
By taking generating series,
one gets:
\ben
&& \sum_{n \geq 1} 2 n \corr{p_2 \cdot p_{n} }_{0} \frac{1}{x_2^{n+1}} \\
& = & \sum_{n \geq 1} \frac{n}{x_2^{n+1}}
\biggl( s (n+1) \cdot \corr{ p_{n+1} }_{0}
 + s(u+v) \cdot \corr{p_1 \cdot p_{n}} _{0} \biggr) \\
& = & s \sum_{n \geq 2} \frac{n(n-1)}{x_2^n} \corr{ p_{n} }_{0}
+ s(u+v) \sum_{n \geq 1} \frac{n}{x_2^{n+1}} \cdot \corr{p_1 \cdot p_{n}}^c_{0}  \\
& = & - s\frac{d}{dx_2} \biggl(x_2^2 \sum_{n \geq 2} \frac{n}{x_2^{n+1}} \corr{ p_{n} }_{0}^c  \biggr)
+ s(u+v) \sum_{n \geq 1} \corr{p_1 \cdot p_{ n}}^c_{0}  \frac{1}{x_2^{n+1}} \\
& = & - s \frac{d}{dx_2} \biggl[ x_2^2 \cdot
\frac{1}{2s} \biggl(1-\frac{s(u+v)}{x_2}
- \sqrt{1-\frac{2s(u+v)}{x_2}+\frac{s^2(u-v)^2}{x_2^2}}\biggr) \biggr]  \\
& + & s(u+v) \cdot
\biggl( \frac{1-\frac{s(u+v)}{x_2}}
{2\sqrt{1-\frac{2s(u+v)}{x_2}+\frac{s^2(u-v)^2}{x_2^2}}} -\frac{1}{2}
\biggr) \\
& = & x_2 \biggl( \frac{1-\frac{s(u+v)}{x_2}
}{\sqrt{1-\frac{2s(u+v)}{x_2}+\frac{s^2(u-v)^2}{x_2^2}}} - 1 \biggr)
- \frac{\frac{2s^2uv}{x_2}}{\sqrt{1-\frac{2s(u+v)}{x_2}+\frac{s^2(u-v)^2}{x_2^2}}}  .
\een
Recall the following combinatorial identity:
\be
\frac{1}{\sqrt{1-2(u+v)z+(u-v)^2z^2}} \\
= \sum_{n=0}^\infty z^{n}
\sum_{k=0}^n \binom{n}{k}^2  u^{n-k} v^k.
\ee
Therefore,
\ben
&& \sum_{n \geq 1} 2n\corr{p_2 \cdot p_{n} }^c_{0}
\frac{1}{x_2^{n+1}}  \\
& = & 2x_2\sum_{n=1}^\infty \frac{s^{n+1}}{x_2^{n+1}}
\sum_{k=1}^n \binom{n}{k} \binom{n}{k-1} u^{n+1-k} v^k \\
& - & \frac{2s^2uv}{x_2}
\sum_{n=0}^\infty \frac{s^n}{x_2^n}
\sum_{k=0}^n \binom{n}{k}^2  u^{n-k} v^k \\
& = & 2\sum_{n=2}^\infty \frac{s^{n+1}}{x_2^{n}}
\sum_{k=1}^n
\biggl[ \binom{n}{k} \binom{n}{k-1}
- \binom{n-1}{k-1}^2
\biggr]u^{n+1-k} v^k,
\een
with the first few terms given by:
\ben
&& \frac{2s^3uv(u+v)}{x_2^2}+\frac{2s^4uv(2u^2+5uv+2v^2)}{x^3}
+ \cdots.
\een

\begin{prop}
The following formula for genus zero dessin two-point function holds:
\be \label{eqn:G02-Des}
\begin{split}
& G_{0,2}(x_1,x_2) \\
=&  \frac{1-\frac{s(u+v)}{x_1}-\frac{s(u+v)}{x_2} + \frac{s^2(u-v)^2}{x_1x_2}}
{2(x_1-x_2)^2\sqrt{(1-\frac{2s(u+v)}{x_1}+ \frac{s^2(u-v)^2}{x_1^2})(1-\frac{2s(u+v)}{x_2}+ \frac{s^2(u-v)^2}{x_2^2})}} \\
-& \frac{1}{2(x_1-x_2)^2}.
\end{split}
\ee
\end{prop}

\begin{proof}
For $m \geq 3$, $n \geq 1$,
\be
\begin{split}
& \frac{m}{s}  \corr{p_{m } \cdot p_n}_{0} \\
= &  (m+n-1) \cdot \corr{ p_{m+n-1}  }_{0}
+ (u+v) (m-1) \corr{p_{m-1} \cdot p_n} _{0} \\
 + & 2 \sum_{k=1}^{m-2}   k\corr{p_{k} \cdot   p_n}_{0}
\cdot (m-1-k) \corr{p_{m-1-k} }_{0},
\end{split}
\ee

Therefore,
\ben
G_{0,2}(x_1, x_2)
& = &  \sum_{m,n \geq 1}\corr{p_ {m} \cdot  p_{n} }^c_{0} \frac{m}{x_1^{m+1}} \frac{n}{x_2^{n+1}} \\
& = & s \sum_{m,n=1}^\infty (m+n-1) \cdot \corr{ p_{m+n-1}  }_{0}
\frac{1}{x_1^{m+1}} \frac{n}{x_2^{n+1}} \\
& + & s \sum_{m\geq 2,n \geq 1} (u+v) (m-1) \corr{p_{m-1} \cdot p_n} _{0}
\frac{1}{x_1^{m+1}} \frac{n}{x_2^{n+1}} \\
& + & 2 s \sum_{m \geq 3,n\geq 1}
\sum_{k=1}^{m-2}   k\corr{p_{k} \cdot  p_n}_{0}
\cdot (m-1-k) \corr{p_{m-1-k} }_{0} \frac{1}{x_1^{m+1}} \frac{n}{x_2^{n+1}} \\
& = & s \sum_{l \geq 1} l \cdot \corr{ p_l }_{0}^c
\sum_{m+n=l+1}\frac{1}{x_1^{m+1}} \frac{n}{x_2^{n+1}}  \\
& + & \frac{s(u+v)}{x_1}G_{0,2}(x_1,x_2)+ 2s G_{0,1}(x_1) \cdot G_{0,2}(x_1,x_2),
\een
so we have:
\ben
G_{0,2}(x_1, x_2)
& = &  s \biggl(1-\frac{s(u+v)}{x_1}-2sG_{0,1}(x_1) \biggr)^{-1}
 \cdot \sum_{l \geq 1}
\corr{ p_{l} }_{0}^c
\sum_{n=1}^{l} \frac{1}{x_1^{l+2-n}}
\frac{n}{x_2^{n+1}}  \\
& = & \frac{s}{\sqrt{1-\frac{s(u+v)}{x_1} + \frac{s^2(u-v)^2}{x_1^2} }}
  \sum_{l \geq 1}
\corr{ p_{l} }_{0}^c
\sum_{n=1}^{l} \frac{1}{x_1^{l+2-n}}
\frac{n}{x_2^{n+1}} .
\een

Furthermore,
\ben
&&  \sum_{l \geq 1}
\corr{ p_{l} }_{0}^c
\sum_{n=1}^{l} \frac{1}{x_1^{l+2-n}}
\frac{n}{x_2^{n+1}} \\
& = &  \frac{1}{x_1(x_1-x_2)^2}
(x_1G_{0,1}(x_1) - x_2 G_{0,1}(x_2))
- \frac{1}{x_1(x_1-x_2)} \frac{d}{dx_2}(x_2 G_{0,1}(x_2)) \\
& = & \frac{1-\frac{s(u+v)}{x_1}-\frac{s(u+v)}{x_2} + \frac{s^2(u-v)^2}{x_1x_2} }
{2s(x_1-x_2)^2\sqrt{(1-\frac{2s(u+v)}{x_2}+ \frac{s^2(u-v)^2}{x_2^2})}} \\
& - & \frac{ \sqrt{(1-\frac{2s(u+v)}{x_1}+ \frac{s^2(u-v)^2}{x_1^2})} }
{2s(x_1-x_2)^2 },
\een
where in the last equality we have used \eqref{eqn:G01}.
Therefore,
we have  proved \eqref{eqn:G02}.
\end{proof}

\subsection{Computation of $n$-point functions in arbitrary genera by Virasoro constraints}

By \eqref{eqn:P2m} we have:
\ben
&& G_{g,n}(x_0, x_1, \dots, x_n) \\
& = & \sum_{m, a_1, \dots, a_n \geq 1}
\corr{p_{m} \cdot p_{a_1} \cdots p_{a_n}}^c_{g}  \frac{m}{x_0^{m+1}}\prod_{i=1}^n \frac{a_i}{x_i^{a_i+1}} \\
& = & s \sum_{m, a_1, \dots, a_n \geq 1}
\sum_{j=1}^n \corr{p_{a_1} \cdots
p_{a_j+m-1} \cdots p_{2a_n}}_{g}^c
\frac{a_j+m-1}{x_0^{m+1}}\prod_{i=1}^n \frac{a_i}{x_i^{a_i+1}} \\
& + & s(u+v) \sum_{m \geq 2, a_1, \dots, a_n \geq 1}
\corr{p_{m-1} \cdot p_{a_1} \cdots p_{a_n} }^c_{0}
\frac{m-1}{x_0^{m+1}}\prod_{i=1}^n \frac{a_i}{x_i^{a_i+1}}\\
& + & s \sum_{k=1}^{m-2} k(m-1-k)\corr{p_{k}p_{m-1-k} \cdot
p_{a_1} \cdots p_{a_n}}^c_{g-1}
\frac{1}{x_0^{m+1}}\prod_{i=1}^n \frac{a_i}{x_i^{a_i+1}} \\
& + & s\sum_{k=1}^{m-2} \sum_{\substack{g_1+g_2=g\\I_1 \coprod I_2 = [n]}}
k \corr{p_{k} \cdot \prod_{i\in I_1} p_{a_i} }^c_{g_1}
\cdot (m-1-k)\corr{p_{m-1-k} \cdot \prod_{i\in I_2} p_{a_i} }^c_{g_2}\\
&& \cdot \frac{1}{x_0^{m+1}}\prod_{i=1}^n \frac{a_i}{x_i^{a_i+1}} .
\een

The following operators are introduced in \cite{Zhou-Even}:
\ben
D_{x_0,x_j}f(x_j)
& = & \frac{x_0f(x_0) -x_j f(x_j)}{x_0(x_0-x_j)^2}
- \frac{1}{x_0(x_0-x_j)} \frac{d}{d x_j}
\big(x_jf(x_j) \big) \\
& = & \frac{\pd}{\pd x_j} \biggl( \frac{x_0f(x_0) - x_jf(x_j)}{x_0(x_0-x_j)}\biggr),
\een
\ben
E_{x_0, u,v} f(u,v) =   \lim_{u \to v} f(u,v)|_{v = x_0}.
\een
With the help these operators,
the dessin Virasoro constraints can be reformulated as follows.

\begin{prop}
Define the renormalized operators $\tilde{D}$ and $\tilde{E}$ as follows:
\begin{align}
\tilde{D}_{x_0,x_j} & = \frac{sD_{x_0,x_j}}{1-2s W_{0,1}(x_0)}  , &
\tilde{E}_{x_0, u,v} & = \frac{sE_{x_0, u,v}}{1-2s W_{0,1}(x_0)}  .
\end{align}
Then one has:
\be \label{eqn:G-g-n}
\begin{split}
&  G_{g,n+1}(x_0, x_1, \dots, x_n) \\
= & \sum_{j=1}^n \tilde{D}_{x_0,x_j}
G_{g,n}(x_1, \dots, x_n)
 + \tilde{E}_{x_0, u,v} G_{g-1,n+2}(u,v, x_1, \dots, x_n) \\
+ & \sum_{\substack{g_1+g_2=g\\I_1 \coprod I_2=[n]}}\;'
\tilde{E}_{x_0, u,v}
\biggl( G_{g_1,|I_1|+1}(u, x_{I_1}) \cdot
G_{g_2,|I_2|+1}(v, x_{I_2}) \biggr).
\end{split}
\ee
\end{prop}

\subsection{Examples}
\label{sec:G03G11} 

We  now present some sample computations  of $G_{g,n}$ using \eqref{eqn:G-g-n}.

\subsubsection{Three-point function in genus zero}

\ben
&&  G_{0,3}(x_0, x_1, x_2) \\
& = & \sum_{j=1}^2 \tilde{D}_{x_0,x_j} G_{0,2}(x_1, x_2)
+ 2 \tilde{E}_{x_0, u,v}
\biggl( G_{0,2}(u, x_1) \cdot G_{0,2}(v, x_2) \biggr) \\
& = & \frac{2s^3uv(1-(u-v)^2s^2(\frac{1}{x_0x_1}+\frac{1}{x_1x_2}+\frac{1}{x_2x_0}) + 2(u+v)(u-v)^2\frac{s^3}{x_0x_1x_2})}
{\prod_{j=0}^2 x_j^2(1-2(u+v)s/x_j+(u-v)^2s^2/x_j^2)^{3/2}}.
\een

\subsubsection{One-point function in genus one}

In this case \eqref{eqn:G-g-n} takes the form:
\ben
G_{1,1}(x_0)
& = &  \tilde{E}_{x_0, u,v} G_{0,2}(u,v)  = \frac{uvs^3}{x_0^4(1- 2(u+v)s/x_0+(u-v)^2/x_0^2)^{5/2} }.
\een
The following are the first few terms:
\ben
G_{1,1}(x_0) & = & \frac{uvs^3}{x_0^4}+\frac{5(u+v)s^4}{x_0^5}
+\frac{(15 u^2+40 u v+15 v^2)s^5}{x_0^6} \\
& +& \frac{35 (u+v) (u^2+4 u v+v^2)s^6}{x_0^7} + \cdots.
\een

\section{The Emergent Geometry of Dessin Tau-Function}
\label{sec:Emergent}

In this Section we show that the dessin Virasoro constraints
implies that the $n$-point functions associated to the dessin tau-function
satisfy the EO topological recursions on a suitable spectral curve.
We take a different approach from that of \cite{Kaz-Zog}.

\subsection{Dessin spectral curve and its special deformation}

For $n \geq 0$,
the dessin Virasoro constraints in genus zero are:
\ben
&&  (u +v)n \frac{\pd F_0}{\pd p_n}
+ \sum_{j=1}^\infty \biggl(p_j- \frac{\delta_{j,1}}{s} \biggr)(n + j) \frac{\pd F_0}{\pd p_{n+ j}} \\
& + &  \sum_{i + j=n} i j \frac{\pd F_0}{\pd p_i} \frac{\pd F_0}{\pd p_j}
+\delta_{n,0}uv  = 0.
\een
Consider the following series:
\be \label{def:Dessin-Special}
y = \frac{1}{2} \sum_{n=1}^\infty
(p_n - \frac{\delta_{n,1}}{s}) x^{n-1}
+ \frac{u+v}{2x} + \sum_{n = 1}^\infty
\frac{n}{x^{n+1}} \frac{\pd F_{0}(t)}{\pd p_n}.
\ee
One has:
\ben
y^2 & = & \frac{1}{4} \bigg(\sum_{n=1}^\infty
(p_n - \frac{\delta_{n,1}}{s}) x^{n-1} \bigg)^2
+ \frac{(u+v)^2}{4x^2} +\sum_{i,j = 1}^\infty
\frac{ij}{x^{i+j+2}} \frac{\pd F_{0}(t)}{\pd p_i} \frac{\pd F_0(t)}{\pd p_j} \\
& + & \frac{u+v}{2} \sum_{n=1}^\infty (p_n - \frac{\delta_{n,1}}{s}) x^{n-2}
+ (u+v) \sum_{n = 1}^\infty
\frac{n}{x^{n+2}} \frac{\pd F_{0}(t)}{\pd p_n} \\
& + & \sum_{j=1}^\infty
(p_j - \frac{\delta_{j,1}}{s})  \sum_{k = 1}^\infty
\frac{k}{x^{k-j+2}} \frac{\pd F_{0}(t)}{\pd p_k}.
\een
Now by Viraroso constraints,
\ben
y^2 & = & \frac{1}{4} \bigg(\sum_{n=1}^\infty
(p_n - \frac{\delta_{n,1}}{s}) x^{n-1} \bigg)^2 + \frac{(u-v)^2}{4x^2}
+ \frac{u+v}{2} \sum_{n=1}^\infty (p_n - \frac{\delta_{n,1}}{s}) x^{n-2} \\
& + & \sum_{k \geq 1} \sum_{j \geq k+1}
(p_j - \frac{\delta_{j,1}}{s})
kx^{j-k-2} \frac{\pd F_{0}(t)}{\pd p_k}.
\een
When $p_n = 0 $ for all $p_n$:
\be
y^2 = \frac{1}{4s^2} - \frac{u+v}{2sx} + \frac{(u-v)^2}{4x^2} .
\ee
We refer to this plane algebraic curve as 
the {\em dessin spectral curve}.
We call \eqref{def:Dessin-Special} the {\em special deformation}
of the dessin spectral curve.

\subsection{Genus zero one-point function and the spectral curve}

Let us take $p_{k} = 0$ in \eqref{def:Dessin-Special} to get:
\be
y=- \frac{1}{2s} + \frac{u+v}{2x} + \sum_{k=1}^\infty \frac{1}{x^{k+1}} 
\frac{\pd F_0}{\pd p_{k}}\biggl|_{p_{n}= 0, n\geq 1}.
\ee
By the definition of correlators and $G_{0,1}(x)$,
\be
y = - \frac{1}{2s} + \frac{u+v}{2x} + G_{0,1}(x).
\ee
By the formula \eqref{eqn:G01} for $G_{0,1}$,
\be
y   = - \frac{1}{2s} \sqrt{1-\frac{2(u+v)s}{x}+ \frac{(u-v)^2s^2}{x^2} },
\ee
and so
\be \label{eqn:Spectral-Des}
y^2 = \frac{1}{4s^2} \biggl(1-\frac{2(u+v)s}{x}+ \frac{(u-v)^2s^2}{x^2} \biggr).
\ee
Note
\be
y= 0 \Leftrightarrow x = s(u+v\pm 2 \sqrt{u}) = s(\sqrt{u} \pm \sqrt{v})^2,
\ee
and so the above equation defines a hyperelliptic rational curve with
a branched covering ramified over $x = s(\sqrt{u} \pm \sqrt{v})^2$.
This is an example of {\em one-cut solution}.
Following \cite{Kaz-Zog},
introduce the following coordinate that globally parameterize the dessin spectral curve:
\be
z = \sqrt{\frac{1-s\beta /x}{1-s\alpha/x}}, \;\;\; \alpha = (\sqrt{u} - \sqrt{v})^2,
\;\; \beta = (\sqrt{u} + \sqrt{v})^2.
\ee
This function in $x$ takes the ramification point $(x= s \beta, y=0)$ to $z=0$,
and $(x= s\alpha, y=0)$ to $z= \infty$.
The functions $x$ and $y$ are two rational functions on this curve,
expressed explicitly in terms of $z$ as follows:
\begin{align}
x(z) & = \frac{s(\alpha z^2 -\beta)}{z^2-1}, &
y(z) & = -\frac{(\alpha-\beta) z}{2s (\alpha z^2  - \beta)}.
\end{align}

There is a natural hyperelliptic structure on this curve: One can define an involution $
p \mapsto \sigma(p)$ by
\be
\sigma(x, y) = (x, -y).
\ee
On the $z$-coordinate,
it is given by
\be
\sigma(z) = -z.
\ee

\subsection{Multilinear differential forms on the dessin spectal curve}

With the introduction of the spectral curve,
one can regard the genus $g$ $n$-point correlation functions
$G_{g,n}(x_1, \dots, x_n)$ as functions on it.
We understand $x$ and $y$ as meromorphic function on the spectral curve.
For a point $p_j$ on it,
we write
\begin{align}
x_j & = x(p_j), & y_j & = y(p_j), & z_j = z(p_j).
\end{align}
Then $G_{g,n}(x_1, \dots, x_n)$ can be written as rational functions in $z_1, \dots, z_n$.
One can also consider the multilinear differential forms:
\be \label{def:Omega}
W_{g,n}(p_1, \dots, p_n)
= \hat{G}_{g,n}(y_1, \dots, y_n) dx_1 \cdots dx_n,
\ee
where $\hat{G}_{g,n}(y_1, \dots, y_n)=G_{g,n}(y_1, \dots, y_n)$ except for the following two exceptional cases:
\bea
&& \hat{G}_{0,1}(y_1) = -\frac{1}{2s}+ \frac{(u+v)}{2x_1} + G_{0,1}(y_1), \\
&& \hat{G}_{0,2}(y_1,y_2) = \frac{1}{(x_1-x_2)^2}  + G_{0,2}(y_1,y_2).
\eea
Since $\hat{G}_{0,1}(y_1) = y_1$, so we have:
\be \label{eqn:Omega-0-1}
W_{0,1}(p_1) = y_1dx_1 = \frac{(\alpha-\beta)^2z_1^2}{(\alpha z_1^2-\beta) (z_1^2-1)^2}dz_1.
\ee
By some straightforward  computations using \eqref{eqn:G02-Des},
one can get:
\be \label{eqn:Omega-0-2}
W_{0,2}(p_1,p_2)
= \frac{dz_1dz_2}{(z_1-z_2)^2}.
\ee
By the results for $G_{0,3}(x_1,x_2,x_3)$ and $G_{1,1}(x_1)$ in \S \ref{sec:G03G11},
we get:
\ben
&& W_{0,3}(p_1, p_2, p_3)
=  G_{0,3}(x_1, x_2, x_3) dx_1dx_2dx_3\\
& = & \frac{2s^3uv(1-(u-v)^2s^2\sum\limits_{j=1}^3
\frac{1}{x_jx_{j+1}}
+ 2(u+v)(u-v)^2\frac{s^3}{x_1x_2x_3})}
{\prod\limits_{j=1}^3 x_j^2(-2sy_j)^3}dx_1dx_2dx_3\\
& = & \frac{1}{(\alpha-\beta)^2} (-\alpha + \frac{\beta}{z_1^2z_2^2z_3^2}) dz_1dz_2dz_3,
\een
where $x_4=x_1$,
and
\ben
&& W_{1,1}(p_1) = G_{1,1}(x_1) dx_1 \\
& = &
\frac{uvs^3}{x_1^4(1- 2(u+v)s/x_1+(u-v)^2/x_1^2)^{5/2} }dx_1
= \frac{uvs^3}{x_1^4(-2sy_1)^5 }dx_1 \\
& = & \frac{1}{8(\alpha-\beta)^2}
\biggl( \frac{\beta}{z^4}- \frac{2\beta+\alpha}{z^2}
+ (2\alpha + \beta) - \alpha z^2 \biggr) dz.
\een

\subsection{Eynard-Orantin topological recursions}

We use $W_{0,2}$ as the Bergman kernel.
Then
\ben
\int_{q=\sigma(p_2)}^{p_2} B(p_1, q)
& = & \int_{z=-z_2}^{z_2} \frac{dz_1dz}{(z_1-z)^2}
 = \frac{dz_1}{z_1-z} \biggl|_{z=-z_2}^{z_2} = \frac{2z_2dz_1}{z_1^2-z_2^2}.
\een
It follows that
\be
K(p_0, p) = \frac{2z dz_0}{2(z_0^2-z^2)\cdot 2ydx}
= \frac{(\alpha z^2-\beta) (z^2-1)^2}{2 (\alpha-\beta)^2z(z_0^2-z^2)} \cdot \frac{dz_0}{dz}.
\ee
This has simple poles at $z=0$ and $z= \pm z_0$.
To understand its behavior at $z = \infty$,
let $z =1/\tilde{z}$ and $z_0 = 1/\tilde{z}_0$.
Then
\be
K(p_0, p) =\frac{(\beta  \tilde{z}^2 - \alpha) (\tilde{z}^2-1)^2}
{2 (\alpha-\beta)^2\tilde{z}(\tilde{z}_0^2-\tilde{z}^2) }
\cdot \frac{d\tilde{z}_0}{d\tilde{z}}.
\ee
This has a simple pole at $\tilde{z} = 0$.

Let us carry out the first few calculations of Eynard-Orantin recursion
for the spectral curve \eqref{eqn:Spectral-Des} with $\omega_{0,1}=W_{0,1}$ and $\omega_{0,2}=W_{0,2}$
given by \eqref{eqn:Omega-0-1} and \eqref{eqn:Omega-0-2} respectively.
\ben
&& \omega_{0,3}(p_0,p_1,p_2) \\
& = & (\Res_{p \to p_+} + \Res_{p \to p_-} ) \biggl(  K(p_0, p)
[W_{0,2}(p, p_1) W_{0,2}(\sigma(p),p_{2}) \\
&& + W_{0,2}(p, p_2) W_{0,2}(\sigma(p), p_1)]  \biggr) \\
& = &  \Res_{z \to 0}  \frac{(\alpha z^2-\beta) (z^2-1)^2}{2 (\alpha-\beta)^2z(z_0^2-z^2)}
\cdot \frac{dz_0}{dz} \\
&& \cdot  \biggl( \frac{dz dz_1}{(z-z_1)^2} \cdot \frac{-dz dz_2}{(-z-z_2)^2}
+ \frac{dzdz_2}{(z-z_2)^2} \cdot \frac{-dzdz_1}{(-z-z_1)^2} \biggr)  \\
& + &  \Res_{\tilde{z} \to 0}
\frac{( \beta \tilde{z}^2-\alpha ) (\tilde{z}^2-1)^2}{2 (\alpha-\beta)^2\tilde{z}(\tilde{z}_0^2-\tilde{z}^2)}
\cdot \frac{d\tilde{z}_0}{d\tilde{z}} \\
&& \cdot \biggl( \frac{d\tilde{z} d\tilde{z}_1}{(\tilde{z}-\tilde{z}_1)^2} \cdot
\frac{-d\tilde{z} d\tilde{z}_2}{(-\tilde{z}-\tilde{z}_2)^2}
+ \frac{d\tilde{z}d\tilde{z}_2}{(\tilde{z}-\tilde{z}_2)^2} \cdot
\frac{-d\tilde{z}d\tilde{z}_1}{(-\tilde{z}-\tilde{z}_1)^2} \biggr)  \\
& = &   \frac{ \beta }{2 (\alpha-\beta)^2  }
\cdot \frac{dz_0}{z_0^2} \cdot  \frac{dz_1}{z_1^2} \cdot \frac{ dz_2}{z_2^2}
+  \frac{\alpha }{2 (\alpha-\beta)^2 }
\cdot \frac{d\tilde{z}_0}{\tilde{z}_0^2}
  \cdot \frac{ d\tilde{z}_1}{ \tilde{z}_1^2} \cdot
\frac{ d\tilde{z}_2}{\tilde{z}_2^2}   \\
& = & \biggl( \frac{ \beta }{(\alpha-\beta)^2  }
\cdot \frac{1}{z_0^2z_1^2z_2^2}
- \frac{\alpha }{(\alpha-\beta)^2 } \biggr) dz_0dz_1dz_2 \\
& = & W_{0,3}(p_0, p_2, p_2).
\een

\ben
\omega_{1,1}(p_0) & = & (\Res_{p\to p_+} + \Res_{p\to p_-})
\biggl( K(p_0, p) W_{0,2}(p, \sigma(p)) \biggr) \\
& = & \Res_{z\to 0} \frac{(\alpha z^2-\beta) (z^2-1)^2}{2 (\alpha-\beta)^2z(z_0^2-z^2)}
\cdot \frac{dz_0}{dz} \cdot \frac{-(dz)^2}{4z^2} \\
& + & \Res_{\tilde{z}\to 0} \frac{(\beta \tz^2-\alpha) (\tz^2-1)^2}{2 (\alpha-\beta)^2\tz(\tz_0^2-\tz^2)}
\cdot \frac{d\tz_0}{d\tz} \cdot \frac{-(d\tz)^2}{4\tz^2} \\
& = & \frac{1}{8(\alpha-\beta)^2}
(\frac{\beta}{z_0^4} - \frac{\alpha+2\beta}{z_0^2} )dz_0
+ \frac{1}{8(\alpha-\beta)^2}
(\frac{\alpha}{\tz_0^4} - \frac{\beta+2\alpha}{\tz_0^2} )d\tz_0\\
& = & \frac{1}{8(\alpha-\beta)^2}
\biggl(\frac{\beta}{z_0^4}
- \frac{\alpha+2\beta}{z_0^2}
+ (\beta+2\alpha) - \alpha z_0^2 \biggr)d z_0\\
& = & W_{1,1}(p_0).
\een

\begin{thm} \label{Thm:Main1}
When $2g-2+n > 0$,
the multi-linear differential forms $W_{g,n}(p_1, \dots, p_n)$
defined by \eqref{def:Omega} has the following form
\be \label{eqn:W-w}
\begin{split}
W_{g,n}(p_1, \dots, p_n)
= & w_{g,n}(z_1, \dots, z_n) dz_1 \cdots dz_n \\
= & \tilde{w}_{g,n}(\tz_1, \dots, \tz_n) d\tz_1\cdots d\tz_n,
\end{split}
\ee
where $w_{g,n}(z_1, \dots, z_n)$  is a Laurent polynomial in $z_1, \dots, z_n$
which has only terms of even degrees with respect to each variable,
similarly for $\tilde{w}_{g,n}(\tz_1, \dots, \tz_n)$.
Furthermore,
they satisfy the Eynard-Orantin topological recursions
given by the spectral curve
\be
y^2 = \frac{1}{4s^2} - \frac{u+v}{2sx} + \frac{(u-v)^2}{4x^2}.
\ee
I.e.,
we have
\be
W_{g,n}(p_1, \dots, p_n)
= \omega_{g,n}(p_1, \dots, p_n).
\ee
\end{thm}

\begin{proof}
We have explicitly check the case of $\omega_{0,3}$ and $\omega_{1,1}$.
We now show that other cases can be checked by induction.
First,
assume the Eynard-Orantin topological recursions are satisfied.
By  the induction hypothesis,
\ben
&& \omega_{g,n+1}(p_0,p_1, \dots, p_n) =
(\Res_{z\to 0} +\Res_{z\to \infty}) K(z_0, z)  \\
&& \cdot  \biggl[ w_{g-1, n+2}(z, -z, z_{[n]}) dzd(-z)   \\
& + & \sum^g_{h=0} \sum_{I \subset [n]}'
w_{h,|I|+1}(z, z_I) w_{g-h, n-|I|+1}(-z, z_{[n]-I}) dzd(-z) \biggr] dz_1 \cdots dz_n.
\een
There are three kinds of terms to consider on the right-hand side.
First,
\ben
&& (\Res_{z\to 0} +\Res_{z\to \infty}) K(z_0, z) \cdot   w_{g-1, n+2}(z, -z, z_{[n]}) dzd(-z) \\
& = & - (\Res_{z\to 0} +\Res_{z\to \infty})
\frac{(\alpha z^2-\beta) (z^2-1)^2}{2 (\alpha-\beta)^2z(z_0^2-z^2)}
\cdot    w_{g-1, n+2}(z, z, z_{[n]}) dz\\
& = &  (\Res_{z\to z_0} +\Res_{z\to -z_0})
\frac{(\alpha z^2-\beta) (z^2-1)^2}{2 (\alpha-\beta)^2z(z_0^2-z^2)}
\cdot    w_{g-1, n+2}(z, z, z_{[n]}) dz \\
& = & - \frac{(\alpha z_0^2-\beta) (z_0^2-1)^2}{2 (\alpha-\beta)^2z_0^2}
\cdot  w_{g-1, n+2}(z_0, z_0, z_{[n]}).
\een
Secondly,
when none of $(h, |I|+1)$ and $(g-h, n-|I|+1)$ is $(0,2)$,
\ben
&& (\Res_{z\to 0} +\Res_{z\to \infty}) K(z_0, z) \cdot
w_{h,|I|+1}(z, z_I) w_{g-h, n-|I|+1}(z, z_{[n]-I}) dzd(-z)   \\
& = &  - (\Res_{z\to 0} +\Res_{z\to \infty})
\frac{(\alpha z^2-\beta) (z^2-1)^2}{2 (\alpha-\beta)^2z(z_0^2-z^2)} \\
&& \cdot  w_{h,|I|+1}(z, z_I) w_{g-h, n-|I|+1}(z, z_{[n]-I})   dz \\
& = &    (\Res_{z\to z_0} +\Res_{z\to -z_0})
\frac{(\alpha z^2-\beta) (z^2-1)^2}{2 (\alpha-\beta)^2z(z_0^2-z^2)} \\
&& \cdot w_{h,|I|+1}(z_0, z_I) w_{g-h, n-|I|+1}(z_0, z_{[n]-I})   dz \\
& = &  -  \frac{(\alpha z_0^2-\beta) (z_0^2-1)^2}{2 (\alpha-\beta)^2z_0^2}
\cdot w_{h,|I|+1}(z_0, z_I) w_{g-h, n-|I|+1}(z_0, z_{[n]-I}) .
\een
Thirdly
we need to consider
\ben
&& (\Res_{z\to 0} +\Res_{z\to \infty}) K(z_0, z) \cdot
(w_{0,2}(z, z_i) w_{g, n}(-z, z_{[n]-\{i\}})  \\
&& + w_{0,2}(-z, z_i) w_{g, n}(z, z_{[n]-\{i\}})) dzd(-z)   \\
& = & dz_0 \cdot (\Res_{z\to 0} +\Res_{z\to \infty})
\frac{(\alpha z^2-\beta) (z^2-1)^2}{2 (\alpha-\beta)^2z(z_0^2-z^2)} \\
&& \cdot \biggl( \frac{1}{(z-z_i)^2} w_{g,n}(z,z_{[n]-\{i\}})
+ \frac{1}{(-z-z_i)^2} \cdot w_{g,n}(z, z_{[n]-\{i\}}) \biggr)dz.
\een
Its computation can be reduced to $dz_0$ times
\ben
&&  (\Res_{z\to 0} - \Res_{z\to \infty})
\frac{(\alpha z^2-\beta) (z^2-1)^2}{2 (\alpha-\beta)^2z(z_0^2-z^2)}
 \cdot  \frac{2(z^2+z_i^2)}{(z^2-z_i^2)^2}  z^{2n} dz,
\een
the result is clear a Laurent series in $z_0$ with only terms of even degrees.
Let us present the result for $n \geq 0$,
the result for $n \geq 0$ is similar.
When $n = 0$,
\ben
&&  (\Res_{z\to 0} - \Res_{z\to \infty})
\frac{(\alpha z^2-\beta) (z^2-1)^2}{2 (\alpha-\beta)^2z(z_0^2-z^2)}
\cdot  \frac{2(z^2+z_i^2)}{(z^2-z_i^2)^2}    dz \\
& = & -\frac{\beta}{(\alpha-\beta)^2z_0^2z_i^2}
+ \frac{(3z_i^2+z_0^2-2)\alpha-\beta}{(\alpha-\beta)^2}.
\een
For $n > 0$,
the residue at $z=0$ vanishes, and so
\ben
&&  (\Res_{z\to 0} - \Res_{z\to \infty})
\frac{(\alpha z^2-\beta) (z^2-1)^2}{2 (\alpha-\beta)^2z(z_0^2-z^2)}
\cdot  \frac{2(z^2+z_i^2)}{(z^2-z_i^2)^2}  z^{2n}  dz \\
& = & -\Res_{w\to 0}
\frac{(\alpha-\beta w^2) (1-w^2)^2}{(\alpha-\beta)^2(z_0^2w^2-1)}
\cdot  \frac{(1+z_i^2w^2)}{(1-z_i^2w^2)^2} w^{-2n-3} dw \\
& = & - \frac{(\alpha-\beta w^2) (1-w^2)^2}{(\alpha-\beta)^2(z_0^2w^2-1)}
\cdot  \frac{(1+z_i^2w^2)}{(1-z_i^2w^2)^2} \biggr|_{w^{2n+2}}.
 \een

On the other hand,
by \eqref{def:Omega},
\ben
&& W_{g,n+1}(y_0, \dots, y_n) \\
& = & G_{g,n+1}(y_0, y_1, \dots, y_n) dx_0 \cdots dx_n  \\
& = & \sum_{j=1}^n \cD_{y_0,y_j} G_{g,n}(y_1, \dots, y_n) \cdot dx_0 \cdots dx_n \\
& + & \cE_{y_0, y,y'} G_{g-1,n+2}(y,y', y_1, \dots, y_n)\cdot  dx_0 \cdots dx_n \\
& + & \sum_{\substack{g_1+g_2=g\\I_1 \coprod I_2=[n]}}\;' \cE_{y_0, y,y'}
\biggl( G_{g_1,|I_1|+1}(y, y_{I_1}) \cdot G_{g_2, |I_2|+1}(y', y_{I_2}) \biggr) \cdot  dx_0 \cdots dx_n .
\een
By comparing these two recursions,
we see that it is easy to show that
\ben
&& (\Res_{z\to 0} + \Res_{z\to \infty}) K(z_0, z)
\cdot  \biggl[ \omega_{g-1, n+2}(z, -z, z_{[n]})  \biggr] \\
& = & \cE_{y_0, y,y'} G_{g-1}(y,y', y_1, \dots, y_n) \cdot dx_0,
\een
and when $(h, |I|+1) \neq (0,2)$ and $(g-h, n-|I|+1) \neq (0,2)$,
\ben
&& (\Res_{z\to 0} + \Res_{z\to \infty}) K(z_0, z)  \biggl[
\omega_{h,|I|+1}(z, z_I) \omega_{g-h, n-|I|+1}(-z, z_{[n]-I})  \biggr] \\
& = & \cE_{y_0, y,y'} \biggl[
G_{h,|I|+1}(y, y_I) G_{g-h, n-|I|+1}(y', y_{[n]-I}) \biggr] \cdot dx_0,
\een
and so it suffices to show that
\ben
&& (\Res_{z\to 0} + \Res_{z\to \infty}) K(z_0, z)  \biggl[
\omega_{0,2}(z, z_i) \cdot \omega_{g, n}(-z, z_1, \dots, \hat{z}_i, \dots,z_n) \biggr] \\
&+ &  (\Res_{z\to 0} + \Res_{z\to \infty}) K(z_0, z)  \biggl[
\omega_{0,2}(-z, z_i) \cdot \omega_{g, n}(z, z_1, \dots, \hat{z}_i, \dots,z_n)  \biggr] \\
& = & \biggl(\cD_{y_0, y_j} \biggl[ G_{g, n}(y_1, \dots, y_n) \biggr] \\
&+ &\cE_{y_0, y,y'} \biggl[G_{0,2}(y,y_i)G_{g,n}(y', y_1, \dots, \hat{y}_i, \dots,y_n) \\
&& \quad \quad + G_{g,n}(y, y_1, \dots, \hat{y}_i, \dots,y_n) G_{0,2}(y',y_i)\biggr] \biggr)\cdot dx_0dx_1\cdots dx_n.
\een
For the computation of the left-hand side,
 we have reduced to  the computation with
$\omega_{g,n}(z, z_1, \dots, \hat{z}_i, \dots, z_n)$ replaced by $z^{2n}dz$,
because
\ben
&& \omega_{g,n}(z, z_1, \dots, \hat{z}_i, \dots, z_n) \\
& = & G_{g,n}(y, y_1, \dots, \hat{y}_i, \dots,y_n) dx dx_1 \cdots \widehat{dx_i} \cdots d x_n,
\een
this corresponds to replacing $G_{g,n}(y,y_1, \dots, \hat{y}_i, \dots, y_N)$
by $f_n(y)$ on the right-hand side, where $f_n(y)$ is chosen such that:
\be
z^{2n} dz = f_n(y) dx,
\ee
i.e.,
\be
f_n(y) = \frac{z^{2n}}{\frac{dx}{dz}} =  - \frac{(z^2-1)^2z^{2n-1}}{2s(\alpha-\beta)},
\ee
so we need to check that:
\ben
&& (\Res_{z\to 0} +\Res_{z\to \infty})
\frac{(\alpha z^2-\beta) (z^2-1)^2}{2 (\alpha-\beta)^2z(z_0^2-z^2)}
 \cdot  \frac{2(z^2+z_i^2)}{(z^2-z_i^2)^2}  z^{2n} dz \\
& = & \biggl(\cD_{y_0, y_i} \biggl[f_n(y_i) \biggr] \\
& + & \cE_{y_0, y,y'} \biggl[G_{0,2}(y,y_i)f_n(y')
 + f_n(y)  \cdot G_{0,2}(y',y_i)\biggr] \biggr) \cdot \frac{dx_0}{dz_0} \frac{d x_i}{dz_i}.
\een
The right-hand can be computed as follows.
Note
\ben
\biggl(G_{0,2}(y_0,y_i) + \frac{1}{(x_0-x_i)^2}\biggr)
\frac{dx_0}{dz_0} \frac{dx_i}{dz_i} = \frac{1}{(z_0-z_i)^2}
\een
and so
\ben
&& \cE_{y_0, y,y'} \biggl[G_{0,2}(y,y_i)f_n(y')
 + f_n(y)   \cdot G_{0,2}(y',y_i)\biggr] \biggr)
 \cdot \frac{dx_0}{dz_0} \frac{d x_i}{dz_i} \\
& = & -\frac{1}{2y_0} \biggl[G_{0,2}(y_0,y_i)f_n(y_0)
 + f_n(y_0)  \cdot G_{0,2}(y_0,y_i)\biggr]  \cdot \frac{dx_0}{dz_0} \frac{d x_i}{dz_i} \\
 & = & -\frac{1}{2y_0} \cdot 2 \biggl(\frac{1}{(z_0-z_i)^2}
\cdot\frac{1}{\frac{\pd x_0}{\pd z_0}\frac{\pd x_j}{\pd z_j}}
- \frac{1}{(x_0-x_i)^2}\biggr)
 \cdot f_n(y_0) \cdot \frac{dx_0}{dz_0} \frac{d x_i}{dz_i} \\
& = & -\frac{1}{2y_0} \cdot 2 \biggl(\frac{1}{(z_0-z_i)^2}
\cdot\frac{1}{\frac{\pd x_0}{\pd z_0}\frac{\pd x_j}{\pd z_j}}
- \frac{1}{(x_0-x_i)^2}\biggr)
 \cdot \frac{z_0^{2n}}{\frac{d x_0}{dz_0}}
\cdot \frac{dx_0}{dz_0}\frac{d x_i}{dz_i} \\
& = & -\frac{1}{2y_0} \cdot
\biggl( \frac{1}{(z_0-z_i)^2} -
\frac{\frac{\pd x_0}{\pd z_0}\frac{\pd x_j}{\pd z_j}}{(x-x_0)^2}
\biggr)\frac{(z_0^2-1)^2}{s(\alpha-\beta)} z_0^{2n-1} \\
& = & \frac{1}{2y_0}
\cdot \frac{1}{(z_0+z_i)^2} \cdot \frac{(z_0^2-1)^2}{s(\alpha-\beta)} z_0^{2n-1} \\
& = & \frac{(z_0^2-1)^2(\alpha z_0^2-\beta)}
{(\alpha-\beta)^2(z_0+z_i)^2}z_0^{2n-2},
\een

on the other hand,
\ben
&& \cD_{y_0, y_i} \biggl[f_n(y_i)  \biggr]  \cdot \frac{dx_0}{dz_0} \frac{d x_i}{dz_i}  \\
& = & -\frac{1}{2y_0} \cdot \frac{\pd}{\pd x_i} \biggl( \frac{x_0f_n(y_0) - x_if_n(y_i)}{x_0(x_0-x_i)}\biggr)
\cdot \frac{dx_0}{dz_0}  \frac{d x_i}{dz_i} \\
& = & - \frac{1}{2y_0} \cdot \frac{\pd}{\pd z_i} \biggl( \frac{x_0f_n(y_0) - x_if_n(y_i)}{x_0(x_0-x_i)}\biggr)
\cdot \frac{dx_0}{dz_0} \\
& = & \frac{s(\alpha z_0^2-\beta)}{(\alpha-\beta)z_0}
\cdot \frac{\pd}{\pd z_i} \biggl(
\frac{-\frac{(\alpha z_0^2-\beta)(z_0^2-1)z_0^{2n-1}}
{2(\alpha-\beta)}
+\frac{(\alpha z_i^2-\beta)(z_i^2-1)z_i^{2n-1}}
{2(\alpha-\beta)}}
{\frac{s(\alpha z_0^2-\beta)}{z_0^2-1}
(\frac{s(\alpha z_0^2-\beta)}{(z_0^2-1)}
-\frac{s(\alpha z_i^2-\beta)}{(z_i^2-1)}) } \biggr) \\
&& \cdot \frac{-2sz_0(\alpha-\beta)}{(z_0^2-1)^2} \\
& = & \frac{1}{(\alpha-\beta)^2}
\cdot \frac{\pd}{\pd z_i} \biggl(
(z_i^2-1) \cdot \frac{ (\alpha z_0^2-\beta)(z_0^2-1)z_0^{2n-1}
-(\alpha z_i^2-\beta)(z_i^2-1)z_i^{2n-1} }
{z_0^2-z_i^2} \biggr) \\
\een
For $n > 0$,
\ben
&&  \frac{(\alpha-\beta w^2) (1-w^2)^2}{(1-z_0^2w^2)}
\cdot  \frac{(1+z_i^2w^2)}{(1-z_i^2w^2)^2} \biggr|_{w^{2n+2}} \\
& = & \frac{(z_0^2-1)^2(\alpha z_0^2-\beta)}
{(z_0+z_i)^2}z_0^{2n-2} \\
& + & \frac{\pd}{\pd z_i} \biggl(
(z_i^2-1) \cdot \frac{ (\alpha z_0^2-\beta)(z_0^2-1)z_0^{2n-1}
-(\alpha z_i^2-\beta)(z_i^2-1)z_i^{2n-1} }
{z_0^2-z_i^2} \biggr).
\een
This can be proved as follows.
The generating function of the left-hand side is
\ben
&&  \sum_{n=1}^\infty u^{2n+2} \cdot \frac{(\alpha-\beta w^2) (1-w^2)^2}{(1-z_0^2w^2)}
\cdot  \frac{(1+z_i^2w^2)}{(1-z_i^2w^2)^2} \biggr|_{w^{2n+2}} \\
& = &  \frac{(\alpha-\beta u^2) (1-u^2)^2}{(1-z_0^2u^2)}
\cdot  \frac{(1+z_i^2u^2)}{(1-z_i^2u^2)^2}
- \alpha
- (3\alpha z_i^2-2\alpha-\beta+\alpha z_0^2) u^2,
\een
and the generating function of the right-hand side is:
\ben
&& \sum_{n=1}^\infty u^{2n+2} \biggl( \frac{(z_0^2-1)^2(\alpha z_0^2-\beta)}
{(z_0+z_i)^2}z_0^{2n-2} \\
& + & \frac{\pd}{\pd z_i} \biggl(
(z_i^2-1) \cdot \frac{ (\alpha z_0^2-\beta)(z_0^2-1)z_0^{2n-1}
-(\alpha z_i^2-\beta)(z_i^2-1)z_i^{2n-1} }
{z_0^2-z_i^2} \biggr) \biggr) \\
& = & \frac{(z_0^2-1)^2(\alpha z_0^2-\beta)}
{(z_0+z_i)^2} \cdot \frac{u^4}{1-z_0^{2}u^2} \\
& + & \frac{\pd}{\pd z_i} \biggl(
(z_i^2-1) \cdot \frac{ (\alpha z_0^2-\beta)(z_0^2-1)
\frac{z_0u^4}{1-z_0^2u^2}
-(\alpha z_i^2-\beta)(z_i^2-1) \frac{z_i u^4}{1-z_i^2u^2} }
{z_0^2-z_i^2} \biggr).
\een
It is straightforward to verify that these generating functions coincide with each other.
This completes the proof.
\end{proof}

\section{Connection to Intersection Numbers on $\Mbar_{g,,n}$}
\label{sec:Intersection}

Eynard \cite{Eynard-Int1, Eynard-Int2} 
related the $n$-point functions constructed from EO topological recursions
to intersection theory.
In this Section we carry out some computations  necessary 
for applying his results to relating Grothendieck's  dessins to intersection numbers
on moduli spaces.

\subsection{Local Airy coordinates}

Recall the dessin spectral curve is given by \eqref{eqn:Spectral-Des} which we recall here:
\ben
&& y^2 = \frac{1}{4s^2} \biggl(1-\frac{2(u+v)s}{x}+ \frac{(u-v)^2s^2}{x^2} \biggr).
\een
This defines a double covering of the complex plane ramified at two points 
\begin{align}
 x_+ & =   s\beta =s(\sqrt{u}+\sqrt{v})^2, & x_- & = s\alpha =s (\sqrt{u}-\sqrt{v})^2,
\end{align}
where
\be
\alpha = (\sqrt{u} - \sqrt{v})^2,
\;\; \beta = (\sqrt{u} + \sqrt{v})^2.
\ee 
Following \cite{Kaz-Zog},
one can  globally parameterize the dessin spectral curve by:
\be
z = \sqrt{\frac{1-s\beta /x}{1-s\alpha/x}}.
\ee
In this coordinate,
the Bergman kernel is given by \eqref{eqn:Omega-0-2} which we recall here:
\ben
W_{0,2}(p_1,p_2)
= \frac{dz_1dz_2}{(z_1-z_2)^2}.
\een

Following \cite{Eynard-Int1, Eynard-Int2},
the local Airy coordinates near $x_\pm$ are defined by:
\be
\xi_\pm = (x-x_\pm)^{1/2}.
\ee

\subsection{Computations of the times}

Plugging $x = x_\pm + \xi_\pm^2$ into the equation of the dessin spectral curve,
one gets:
\be
y^2 = \pm \frac{\xi_\pm^2 (4 \sqrt{uv} s \pm \xi_\pm^2)}{(s(\sqrt{u}\pm \sqrt{v})^2 + \xi_\pm^2)^2}.
\ee
So we get:
\be
y = \xi_\pm \cdot \frac{2(\sqrt{u})^{1/2}(\pm \sqrt{v})^{1/2}s^{1/2}}{s(\sqrt{u}\pm\sqrt{v})} \cdot
\frac{(1\pm \frac{\xi_\pm^2}{4 \sqrt{uv} s})^{1/2}}{1+ \frac{\xi^2_\pm}{s(\sqrt{u}\pm \sqrt{v})^2}}
\cdot 
\ee
It is straightforward to write down the series expansions in $\xi_\pm$ using
\ben
&& (1+4x)^{1/2} = 1 + 2 \sum_{m=0}^\infty \frac{(-1)^m}{m+1}\binom{2m}{m} x^{m+1}, \\
&& \frac{1}{1+x} = \sum_{n=0}^\infty (-1)^n x^n.
\een

\subsection{The expansions of the Bergman kernel in local Airy coordinates}

Near $p_1=p_2=p_+=(x_+, 0)$, 
we have
\ben
z_j = \sqrt{\frac{1-s(a+b)^2 /(s(a+b)^2+\xi_{j,+}^2)}{1-s(a-b)^2/(s(a+b)^2+\xi_{j,+}^2)}}
= \frac{\xi_{j, +}}{(4sab+ \xi_{j,+}^2)^{1/2}},
\een
where $a= \sqrt{u}$, $b=\sqrt{v}$,
and so 
\ben
B(z_1, z_2) & = & \frac{dz_1dz_2}{(z_1-z_2)^2} \\
& = & \frac{d\xi_{1,+}}{(1+\frac{\xi_{1,+}^2}{4abs})^{1/2}} 
\cdot \frac{d\xi_{2,+}}{(1+\frac{\xi_{2,+}^2}{4abs})^{1/2}} \\
&& \cdot \frac{1}{(\xi_{1,+}(1+\frac{\xi_{2,+}^2}{4abs})^{1/2}-\xi_{2,+}(1+\frac{\xi_{1,+}^2}{4abs})^{1/2})^2}.
\een 
Its expansion can be computed using the combinatorial identity:
\ben
&& \frac{1}{(1+4t x^2)^{1/2}(1+ 4t y^2)^{1/2}(x(1+ 4t y^2)^{1/2}-y(1+ 4t x^2)^{1/2})^2} \\
& = & \frac{1}{(x-y)^2}
+ \sum_{n\geq 0} (n+2) (-t)^{n+1} \sum_{k=0}^n 2 \frac{\binom{n}{k}^2 \binom{2n+2}{n}}{\binom{2n+2}{2k+1}} x^{2k}y^{2n-2k}.
\een
This is found and proved as follows.
With the help of Maple we get:
\ben
&& \frac{1}{(1+4t x^2)^{1/2}(1+ 4t y^2)^{1/2}(x(1+ 4t y^2)^{1/2}-y(1+ 4t x^2)^{1/2})^2} \\
& = & \frac{1}{(x-y)^2} -2t + 3(2x^2+2y^2)t^2
- 4(5x^4+6x^2y^2+5y^4)t^3 + \cdots,
\een
After consulting with \cite{OEIS}
we find the numbers $1, 2,2,5,6,5,\dots$ are the integer sequence 
A120406 on \cite{OEIS},
they are given by the closed formula:
\be
T(n,k) = 2 \frac{\binom{n}{k}^2 \binom{2n+2}{n}}{\binom{2n+2}{2k+1}},
\ee
and have the following generating series: 
\be
\begin{split}
& 1-\sqrt{(1-ax)(1-bx)} \\ 
= & \frac{(a+b)x}{2} 
 + \frac{(b-a)^2x^2}{8}\sum_{n \geq 0}( \sum_{0\leq k \leq n}T(n,k)a^kb^{n-k} )(\frac{x}{4})^n.
\end{split}
\ee
From this we get:
\ben 
&& 1-\sqrt{(1+4tx^2)(1+4ty^2)} \\
&= & - 2(x^2+y^2)t
 + 2(x^2-y^2)^2 \sum_{n \geq 0}(-1)^nt^{n+2} \sum_{0\leq k \leq n}T(n,k)x^{2k}y^{2n-2k} .
\een
After taking $\frac{\pd}{\pd t}$ on both sides we get:
\ben
&& \sum_{n\geq 0} (n+2) (-t)^{n+1} \sum_{k=0}^n 2 \frac{\binom{n}{k}^2 \binom{2n+2}{n}}{\binom{2n+2}{2k+1}} x^{2k}y^{2n-2k} \\
& = & \frac{x^2+y^2 + 8tx^2y^2 - (x^2+y^2)(1+4tx^2)^{1/2}(1+4ty^2)^{1/2}}{(x^2-y^2)^2
(1+4tx^2)^{1/2}(1+4ty^2)^{1/2}}.
\een
One can then proceeds as follows:
\ben
&& \frac{1}{(1+4t x^2)^{1/2}(1+ 4t y^2)^{1/2}(x(1+ 4t y^2)^{1/2}-y(1+ 4t x^2)^{1/2})^2} \\
& = & \frac{(x(1+ 4t y^2)^{1/2}+y(1+ 4t x^2)^{1/2})^2}{(1+4t x^2)^{1/2}(1+ 4t y^2)^{1/2}(x^2-y^2)^2} \\
& = & \frac{x^2+y^2 + 8tx^2y^2 + 2xy(1+4tx^2)^{1/2}(1+4ty^2)^{1/2}}{(x^2-y^2)^2
(1+4tx^2)^{1/2}(1+4ty^2)^{1/2}} \\
& = & \frac{1}{(x-y)^2}
+ \sum_{n\geq 0} (n+2) (-t)^{n+1} \sum_{k=0}^n 2 \frac{\binom{n}{k}^2 \binom{2n+2}{n}}{\binom{2n+2}{2k+1}} x^{2k}y^{2n-2k}.
\een

Near $p_1=p_2=p_-=(x_-, 0)$,
we have
\ben
z_j = \sqrt{\frac{1-s(a+b)^2 /(s(a-b)^2+\xi_{j,-}^2)}{1-s(a-b)^2/(s(a-b)^2+\xi_{j,-}^2)}}
= \frac{(-4sab+ \xi_{j,-}^2)^{1/2}}{\xi_{j, -}}
\een 
and so
\ben
B(z_1, z_2)  
& = & \frac{d\xi_{1,-}}{(1+\frac{\xi_{1,-}^2}{4abs})^{1/2}}
\cdot \frac{d\xi_{2,-}}{(1+\frac{\xi_{2,-}^2}{4abs})^{1/2}} \\
&& \cdot \frac{1}{(\xi_{1,-}(1+\frac{\xi_{2,-}^2}{4abs})^{1/2}-\xi_{2,-}(1+\frac{\xi_{1,-}^2}{4abs})^{1/2})^2}.
\een
Its expansion can be found in the same fashion.

Near Near $p_1=p_+=(x_+,0)$ and $p_2=p_-=(x_-, 0)$,
we have
\begin{align*}
z_1 & = \frac{\xi_{1, +}}{(4sab+ \xi_{1,+}^2)^{1/2}}, & 
z_2 & = \frac{(-4sab+ \xi_{2,-}^2)^{1/2}}{\xi_{2, -}},
\end{align*}
In this case 
\ben
B(z_1, z_2) & = & \frac{\sqrt{-1}d \xi_{1,+} d\xi_{2,-}}
{4abs (1+\frac{\xi_{1,+}^2}{4abs})^{1/2}(1-\frac{\xi_{2,-}^2}{4abs})^{1/2} \biggl(
(1+\frac{\xi_{1,+}^2}{4abs})^{1/2}(1-\frac{\xi_{2,-}^2}{4abs})^{1/2} 
- \frac{\xi_{1,+}\xi_{2,-}}{\sqrt{-1}4abs}\biggr)^2}.
\een
It is possible to find explicit expressions for its expansion by noting:
\ben
&& \frac{1}{(1+4tx^2)^{1/2}(1+4ty^2)^{1/2}((1+4tx^2)^{1/2}(1+4ty^2)^{1/2} + 4txy)^2} \\
& = & \frac{((1+4tx^2)^{1/2}(1+4ty^2)^{1/2}- 4txy)^2}{(1+4tx^2)^{1/2}(1+4ty^2)^{1/2}(1+4tx^2+4ty^2)^2} \\
& = & \frac{(1+4tx^2)^{1/2}(1+4ty^2)^{1/2}}{(1+4tx^2+4ty^2)^2}  \\
& - & \frac{8txy^2}{(1+4tx^2+4ty^2)^2}  \\
& + & \frac{16t^2x^2y^2}{(1+4tx^2)^{1/2}(1+4ty^2)^{1/2}(1+4tx^2+4ty^2)^2}. 
\een
The Taylor expansion of the right-hand side of the second equality can be found by standard method.

\section{Concluding Remarks}
\label{sec:Conclusion}

In this paper we have discussed the emergent geometry of the enumeration 
of Grothendieck's dessins d'enfant.
This leads to Eynard-Orantin topological recursion on the spectral curve
\be
y^2 = \frac{1}{4s^2} \biggl(1-\frac{2(u+v)s}{x}+ \frac{(u-v)^2s^2}{x^2} \biggr).
\ee 
This curve can be found by computing the genus zero one-point function of the dessin tau-function.
Furthermore,
by computing the genus zero two-point function of the dessin tau-function,
the Bergman kernel is shown to be
\be
B(z_1, z_2) = \frac{dz_1dz_2}{(z_1-z_2)^2},
\ee 
where 
\be
z_j = \sqrt{\frac{x_j-s(\sqrt{u}+\sqrt{v})^2}{x_j- s(\sqrt{u}-\sqrt{v})^2}}.
\ee

We have compared with the genus zero one-point functions of some related theories.
It turns out that in the cases of Witten-Kontsevich tau-function (with nonzero $t_0$),
partition function of  Hermitian one-matrix models 
and modified partition function  of Hermitian one-matrix models with even couplings
(after the use of t'Hooft coupling constant),
the genus zero one-point functions are all related to Catalan numbers,
but for the dessin partition function,
it is related to Narayana numbers of type A:
\be
N(A_n, q) = \sum_{k=0}^n \frac{1}{n+1} \binom{n+1}{k} \binom{n+1}{k+1} q^k.
\ee
More precisely,
\be
\begin{split}
G_{0,1}(x) = &  \sum_{n=1}^\infty \frac{s^n}{x^{n+1}} 
\sum_{k=1}^n \frac{1}{n} \binom{n}{k}\binom{n}{k-1} u^{n+1-k}v^k \\
= & \frac{1}{2s}
\biggl(1-\frac{s(u+v)}{x}
- \sqrt{1-\frac{2s(u+v)}{x}+\frac{s^2(u-v)^2}{x^2}}\biggr).
\end{split}
\ee
The generalized Narayana numbers of type B/C are given by:
\be
N(B_n,q) = \sum_{k=0}^n \binom{n}{k}^2 q^k,
\ee
and the generaized Narayana numbers of type D are given by:
\be
N(D_n,q) = 1+q^n  +\sum_{k=1}^{n-1} 
\biggl[\binom{n}{k}^2 - \frac{n}{n-1} \binom{n-1}{k-1} \binom{n-1}{k} \biggr] q^k.
\ee

From the generating series 
\be
\sum_{n=0}^\infty x^n \sum_{k=0}^n \binom{n}{k}^2 y^k
= \frac{1}{\sqrt{1-2x-2xy+x^2-2x^2y+x^2y^2}}
\ee
we get:
\ben
\sum_{n=0}^\infty \frac{s^n}{x^{n+1}} \sum_{k=0}^n \binom{n}{k}^2 u^{n-k} v^k 
=\frac{1}{\sqrt{x^2-2s(u+v)x+s^2(u-v)^2}},
\een 
and so we conjecture that
\be
y^2 = \frac{1}{x^2-2s(u+v)x+s^2(u-v)^2}
\ee
is related to the spectral curve for enumeration of ``Grothendieck's dessins of Type B/C".
Similarly,
for Type D,
\ben
&& \sum_{n=0}^\infty  \frac{s^n}{x^{n+1}} \biggl(
u^n+v^n  +\sum_{k=1}^{n-1}
\biggl[\binom{n}{k}^2 - \frac{n}{n-1} \binom{n-1}{k-1} \binom{n-1}{k} \biggr] u^{n-k}v^k \biggr) \\
& = & \sum_{n=0}^\infty \frac{s^n}{x^{n+1}} \sum_{k=0}^n \binom{n}{k}^2 u^{n-k} v^k  \\
& - & \sum_{n=1}^\infty  \frac{(n+1)s^{n+1}}{x^{n+2}}  
\sum_{k=1}^{n} \frac{1}{n} \binom{n}{k-1} \binom{n}{k}  u^{n+1-k}v^k \\
& = & \frac{1}{\sqrt{x^2-2s(u+v)x+s^2(u-v)^2}} \\
& + & s \frac{\pd}{\pd x} \sum_{n=1}^\infty \frac{s^n}{x^{n+1}}
\sum_{k=1}^n \frac{1}{n} \binom{n}{k}\binom{n}{k-1} u^{n+1-k}v^k \\
& = & \frac{s(u+v)}{2x^2}
+ \frac{2-\frac{s(u+v)}{x}+\frac{s^2(u-v)^2}{x^2}}{2x \sqrt{1-\frac{2s(u+v)}{x}+\frac{s^2(u-v)^2}{x^2} }}.
\een
We conjecture that it
is related to the spectral curve for enumeration of ``Grothendieck's dessins of Type D".
One can also take $u=v$ in the above discussions.

The Narayana numbers and Catalan numbers have a wide connection to a plethora of mathematical objects,
including combinatorics of Coxeter groups, generalized noncrossing partitions,
free probability, cluster algebras, etc. (See \cite{Fom-Rea} and \cite{Arm} for expositions.)  
We believe the surprising appearance of these numbers in the emergent geometry of Grothendieck's dessins
and clean dessins suggest deeper connections to these objects that deserve further investigations.

\vspace{.2in}
{\bf Acknowledgements}.
The author is partly supported by NSFC grants 11661131005 and 11890662.
The numerical computations in this paper are carried out with the help of Maple and \cite{OEIS}.

 \bibliographystyle{plain}

\end{document}